\documentclass[sigconf, preprint, nonacm]{acmart}
\setcopyright{none}
\usepackage{algorithm}
\usepackage{algorithmicx}
\usepackage{algpseudocode}
\usepackage{dblfloatfix}

\usepackage{amsthm}
\newtheorem{theorem}{Theorem}[section]

\newcommand\vldbavailabilityurl{URL_TO_YOUR_ARTIFACTS}
 
\begin{document}
\title{AutoCSF: Provably Space-Efficient Indexing of Skewed Key-Value Workloads via Filter-Augmented Compressed Static Functions}

\author{David Torres Ramos}
\affiliation{%
  \institution{Distill}
  \city{New York}
  \country{USA}
}
\email{detorresramos1@gmail.com}

\author{Vihan Lakshman}
\affiliation{%
  \institution{MIT CSAIL}
  \city{Cambridge}
  \country{USA}
}
\email{vihan@mit.edu}

\author{Chen Luo}
\affiliation{%
  \institution{Amazon}
  \city{Palo Alto}
  \country{USA}
}
\email{cheluo@amazon.com}

\author{Todd Treangen}
\affiliation{%
  \institution{Rice University}
  \city{Houston}
  \country{USA}
}
\email{treangen@rice.edu}

\author{Benjamin Coleman}
\affiliation{%
  \institution{Google DeepMind}
  \city{Mountain View}
  \country{USA}
}
\email{colemanben@google.com}

\begin{abstract}
We study the problem of building space-efficient, in-memory indexes for massive key-value datasets with highly skewed value distributions. This challenge arises in many data-intensive domains and is particularly acute in computational genomics, where $k$-mer count tables can contain billions of entries dominated by a single frequent value. While recent work has proposed to address this problem by augmenting \emph{compressed static functions} (CSFs) with pre-filters, existing approaches rely on complex heuristics and lack formal guarantees. In this paper, we introduce a principled algorithm, called AutoCSF, for combining CSFs with pre-filtering to provably handle skewed distributions with near-optimal space usage. We improve upon prior CSF pre-filtering constructions by (1) deriving a mathematically rigorous decision criterion for when filter augmentation is beneficial; (2) presenting a general algorithmic framework for integrating CSFs with modern set membership data structures beyond the classic Bloom filter; and (3) establishing theoretical guarantees on the overall space usage of the resulting indexes. Our open-source implementation of AutoCSF demonstrates space savings over baseline methods while maintaining low query latency.
\end{abstract}
\setcopyright{none}
\maketitle



\section{Introduction}


Key-value stores are indispensable components of modern data infrastructure. As data volumes continue to balloon in size, there is increasingly a need to design new data structures for key-value lookups with small memory footprints while still enabling fast queries. When the set of key-value pairs is fixed (or changing infrequently), \emph{succinct retrieval} data structures have emerged as a state-of-the-art class of techniques for indexing data with minimal space overhead. These retrieval structures offer the guarantee of returning the correct value when queried with any key in the index, but may return arbitrary outputs for any key not in the set. 

Amongst these succinct retrieval methods, \emph{compressed static functions} (CSFs) achieve space usage proportional to the empirical entropy of the value set. Thus, when the list of values in the index is highly repetitive, CSFs can achieve significantly improved compression rates compared to distribution-agnostic methods, typically within a constant factor of the information-theoretic lower bound.

However, one limitation of CSFs is that, by construction, the index must use a lower bound of at least 1 bit per key. Consequently, if the value distribution is highly skewed, namely when a single value is associated with the majority of the keys, CSFs fall short of achieving optimal compression rates. This barrier is especially prominent in computational genomics where long read $k$-mer count tables, which map keys of $k$-length nucleotide substrings to their frequency of occurrence in a given genome, often exhibit extreme skew with over 90\% of $k$-mers having a value of 1. 

This challenge was previously identified in the literature by \citet{shibuya2022space} who proposed a new index called a \textit{Bloom-enhanced Compressed Static Function} (BCSF) which augments a CSF with a Bloom filter to handle the dominating value. In particular, a BCSF index first inserts all keys not associated with the dominating value into a Bloom filter.  Additionally, since Bloom filters admit false positives, the BCSF index will insert the false positive keys associated with the dominating value into the CSF. At query time, given a key $k$, a BCSF first checks if $k$ is in the Bloom filter. If the filter returns false, then we can conclude that the associated value must be the dominating entity and simply return that value. If the filter returns true, we query the CSF and return the output. This construction guarantees correctness while eliminating the need to index the dominating value directly. Inspired in part by this work, \citet{hermann2025learned} also leverage this idea of filter augmentation to develop a \emph{learned} CSF index that allows for improved compression via machine learning.

Despite the effectiveness of filter-augmented CSFs, as demonstrated in these prior works in the literature, these data structure compositions introduce additional complexity. In particular, it is not completely clear which key-value distributions might benefit from filter-augmented CSFs or how to set the false positive rate of the filter to minimize the overall space usage. Unlike standalone approximate set membership data structures, which maintain a monotonic relationship between the false positive rate and the size of the resulting filter, the relationship between the false positive rate and the overall index size in filter-augmented CSFs is complex to model. As an illustrative example, increasing the false positive rate of a filter might allow more values to pass through to the CSF, causing the value frequency distribution to be more closely-aligned to powers of 2 and thus more highly compressible by Huffman codes. While the work of \cite{shibuya2022space} present a set of criteria for making these parameter choices in a BCSF, they are heuristic-driven as opposed to mathematically principled and, as we will demonstrate in this paper, can be misleading and result in suboptimal decisions. Motivated by these shortcomings, we ask the following research questions:

\begin{enumerate}
\item Can we develop a mathematically principled decision criterion to determine when a set of key-value pairs would benefit from filter-augmentation to save space versus simply using only a CSF? 

\item Given a filter-augmented CSF, can we design an efficient algorithm for determining the optimal parameters for the pre-filter that minimizes the total space usage of the index? 

\item Can we design a general framework for Questions (1) and (2) that can be applied to all set membership data structures in the literature, as opposed to just the classical Bloom filter? 
\end{enumerate}

In this paper, we affirmatively answer all three of these research questions through a \emph{single} algorithm, which we call AutoCSF. The key insight behind AutoCSF is that while modeling the exact memory cost of a filter-augmented CSF appears to be intractable -- or at least unwieldy to analyze -- we can derive analytically clean upper and lower bounds on the \emph{difference} in space usage between a CSF and filter-augmented CSF. In other words, instead of attempting to develop separate cost models for the space usage of a CSF and filter-augmented CSF, we bound the difference in bits per key between the two data structures \emph{without ever having to explicitly construct either index}. This observation drives the development of our AutoCSF algorithm, which first determines if filter-augmentation is beneficial and then, if a filter is helpful, provides near-optimal settings for the filter parameters. The AutoCSF framework can be applied to any approximate set membership data structure with a known relationship between the memory usage and false positive rate, allowing our results to easily extend to the large number of membership filters that have been developed over the last 50 years~\cite{bloom1970space,mitzenmacher2001compressed,almeida2007scalable,porat2009optimal,fan2014cuckoo,pellow2017improving,mitzenmacher2018model,vaidya2020partitioned,liu2020stable,graf2020xor,mitzenmacher2020adaptive,BuRR2022,sato2023fast}. We implement AutoCSF in practice and verify that our theoretical results align closely with empirical results (i.e. the bounds are tight), thus providing a principled algorithm to effectively compose CSFs with prefilters. 

In summary, we make the following contributions

\begin{itemize}
\item We demonstrate that the existing state-of-the-art heuristic-based decision criterion for filter-augmented CSFs (BCSF) can actively degrade performance, recommending filter configurations that increase space usage by up to 1 bit per key relative to a plain CSF (Figure \ref{fig:shibuya_bpk}).

\item We derive lower and upper bounds on the difference in space usage between a standard CSF and a filter-augmented CSF (Theorems \ref{thm:lower_bound} and \ref{thm:upper-bound}). The key insight is to bound the cost difference directly using Huffman code optimality, rather than attempting to model the cost of each index independently. The lower bound yields a provably safe decision criterion: it never recommends a filter that increases space.

\item We present a general algorithmic framework, AutoCSF, for determining \emph{when} to use a filter-augmented CSF and \emph{how} to set the parameters for the prefilter for space efficiency. AutoCSF applies to any approximate set membership data structure with a known cost-vs-false-positive-rate tradeoff, including XOR filters, binary fuse filters, and Bloom filters with arbitrary hash counts. This extends the filter-augmented CSF paradigm beyond the classical Bloom filter used in prior work. In addition, unlike prior approaches, AutoCSF does not require estimating the full value distribution or entropy. Instead, it relies only on the dominant value fraction and a filter cost model, making it both simpler and more robust in practice.

\item We empirically validate \footnote{Our experimetal code is available at \url{https://github.com/detorresramos/CaramelDB-Benchmarks}.} that our bounds are tight across three synthetic distributions and five filter types, that the theory-guided parameter selection matches or comes within one discrete step of the empirical optimum, and that AutoCSF achieves Pareto-optimal memory-latency tradeoffs on real genomics workloads with up to 70 million keys.
\end{itemize}


\section{Related Work}

Succinct retrieval data structures, including CSFs, have garnered intense research interest in the theoretical computer science community for decades. The original CSF constructions involved using a minimal-perfect hash function to index a set of constant-size values, later followed by algorithms that store Huffman codes of the values, and finally culminating in algorithms that \textit{compute} the Huffman codes of the values through the solution of a randomized binary linear system~\cite{majewski1996family,dietzfelbinger2008succinct,belazzougui2013compressed}.

This binary, XOR-based system undergoes a phase transition as we reduce the space, and there has been considerable theoretical work to migrate from the easier, larger-overhead regime~\cite{botelho2013practical} (where hypergraph peeling solves the system in linear time with high probability) to the minimum-space regime where the system is not easily solvable~\cite{dietzfelbinger2008succinct}.
This minimum-space regime was primarily considered a theoretical curiosity and not feasible in practice due to a construction time cubic in the number of keys. The work of \cite{genuzio2020fast} refuted this conventional wisdom by presenting a practical CSF construction algorithm that leverages both sophisticated algorithmic techniques, such as lazy Gaussian elimination and hypergraph peeling, as well as implementation-level optimizations like broadword programming. 

Following the pioneering work of \citet{genuzio2020fast}, which provided the first practically implementable CSF construction, several works have proposed building data systems with CSFs as the foundational indexing primitive \cite{coleman2023caramel, shibuya2022space, hermann2025learned}. The most relevant prior work to our contributions in this paper is that of \citet{shibuya2022space} who proposed the idea of Bloom-enhanced CSFs to handle skewed value distributions, motivated by the problem of indexing skewed $k$mer count datasets in computational genomics. Our contributions in this paper directly build off the ideas in \cite{shibuya2022space}. Specifically, as we discuss further in the next section, the heuristic-driven decision criteria used in \cite{shibuya2022space} to determine when to apply a Bloom filter and how to set the false positive rate can lead to suboptimal configurations in practice. We improve upon this construction with our AutoCSF algorithm. 

In addition to the work of \citet{shibuya2022space}, several other recent works have studied CSFs from a practical lens. CARAMEL \cite{coleman2023caramel} explores indexing a multi-set of values for each key via an array of CSFs. More recently, \citet{hermann2025learned} proposed a learned CSF (LCSF) that combines tools from the rich literature on learned index structures \cite{kraska2018case} as well as Bloom filter augmentations to reduce the memory footprint of CSFs even further, though at the cost of increased query latency and long model training times. In this paper, we include an experimental comparison of our proposed AutoCSF framework with LCSF on skewed distributions where we find that we can in fact achieve comparable space savings without the overhead of training a learned model. This result may be of independent interest in demonstrating that learned indexes, despite their power, may provide limited utility in settings where we can analytically model the cost of an index directly. 

\section{Warmup: The BCSF Algorithm}

In this section, we review the BCSF heuristic introduced by \cite{shibuya2022space} as the design of this algorithm will motivate our efforts to improve upon it in this paper. Let $K$ denote the set of all items in the dataset and let $K_0$ denote the set of items with the most common frequency. We have that $|K_0| = \alpha |K|$ for some $0 \le \alpha \le 1$. Our goal is to minimize space by storing the set of non-dominating values in a Bloom filter and the associated counts of these non-dominating values in a CSF (plus the false positives).

Assume that the Bloom filter takes $C_{BF} \log \frac{1}{\varepsilon}$ bits per key where $\varepsilon$ denotes the false positive rate. Also assume the CSF takes $C_{CSF}$ bits per key. The total space usage of this Bloom-enhanced CSF data structure is thus $$C_{BF} (1-\alpha) |K| \log \frac{1}{\varepsilon} + C_{CSF} |K| ((1-\alpha) + \varepsilon \alpha)$$ where we assume the logarithm is base 2. \\

To determine if we need a Bloom filter, BCSF needs to verify if the following inequality holds $$C_{BF} (1-\alpha) |K| \log \frac{1}{\varepsilon} + C_{CSF} |K| ((1-\alpha) + \varepsilon \alpha) < C_{CSF} |K|$$

We can rewrite this inequality as $$\frac{C_{BF}}{C_{CSF}} \left(\frac{1-\alpha}{\alpha}\right) \log \frac{1}{\varepsilon} + \varepsilon < 1$$

The left-hand side of the previous inequality reaches its minimum at $$\varepsilon^* = \frac{C_{BF}}{C_{CSF}} \left(\frac{1-\alpha}{\alpha}\right) \log e$$

This gives us our decision criteria. We conclude that the Bloom filter helps reduce space if $\varepsilon^* < 1$ and $$\alpha > \frac{C_{BF} \log e}{C_{CSF} + C_{BF}\log e}$$

If $\alpha$ is sufficiently large, we can then construct a Bloom filter with false positive rate $\varepsilon^*$. 

The question remains: how do we compute $C_{BF}$ and $C_{CSF}$, the bits-per-key cost of the Bloom filter and the CSF? The BCSF algorithm estimates $C_{BF}$ by the theoretical space coefficient of idealized Bloom filters of 1.44. However, this idealized model is not realistic because it assumes the number of hash functions used by filter is continuous as opposed to a discrete parameter, which is the case in practice. In this paper, we explore whether we can do better by having a more accurate Bloom filter cost model. Secondly, the BCSF algorithm, as described in \cite{shibuya2022space}, estimates $C_{CSF}$ in terms of the empirical entropy $H_0$, producing the piecewise equation 

\[
C_{CSF} = 
\begin{cases} 
0.22H_0^2 + 0.18H_0 + 1.16, \quad \text{if } H_0 < 2 \\

1.1H_0 + 0.2 \quad \text{otherwise}
\end{cases}
\]

However, this CSF cost model is a heuristic and not theoretically principled. In Figure \ref{fig:shibuya_bpk} we demonstrate that this heuristic can provide misleading guidance in practice, which motivates our work to design a more principled algorithm with provable guarantees with AutoCSF. The key challenge in designing a principled algorithm is that the cost of the CSF space usage, particularly when coupled with random false positives from the filter, is difficult to model analytically, which is why a data-driven heuristic as in BCSF is a natural choice. However, as we demonstrate in the next section, we can circumvent this difficulty by bounding the \emph{difference} in cost between a CSF and filter-enhanced CSF as opposed to modeling the cost of each index individually. This insight allows us to derive tight, usable bounds on the bit per key. 

\section{The AutoCSF Algorithm}

We describe AutoCSF, our new algorithm for filter-enhanced CSFs in Algorithm \ref{alg:autocsf}. The algorithm is responsible for making two key design decisions: (1) first determine if a given set of key-value pairs would achieve more compression with filter-augmentation as opposed to just a standard CSF; and (2) if filter-augmentation helps, set the filter parameters near-optimally. Algorithm \ref{alg:autocsf} is parameterized by a filter bits-per-key cost function $b(\varepsilon)$ and can thus be applied to any approximate set membership data structure, including recent powerful techniques like XOR filters \cite{graf2020xor} and binary fuse filters \cite{graf2022binary} that provide a closed-form cost function $b(\varepsilon)$. The crucial step in Algorithm \ref{alg:autocsf} is to optimize a lower bound on the difference in cost between a CSF and filter-enhanced CSF (which we denote as FSF). 

By maximizing this lower bound in terms of the filter false positive rate $\varepsilon$, Algorithm \ref{alg:autocsf} aims to find a parameter setting that maximizes the space savings of using a filter-augmented CSF over a regular CSF. If these savings, even after maximization, are negative, then we conclude that a filter does not help. Otherwise, we build a filter-augmented CSF using the optimal false positive rate $\epsilon^*$. In the next subsection, we formally prove this lower bound on the bits per key difference in cost between the respective indexes. Then, in the following subsection, we prove an analogous upper bound. Although we do not use the upper bound directly in the AutoCSF algorithm, it is nevertheless a useful theoretical result in establishing the first such bound for this problem. In addition, the upper bound can be used for predicting the maximum possible space savings from a filter-augmented CSF, which can be used for making cost-effective decisions on infrastructure to serve these indexes. 

\begin{algorithm}
\caption{AutoCSF}
\begin{algorithmic}[1]
\Require A set of key--value pairs $\{(k_i, v_i)\}_{i=1}^N$ with $n$ unique values, most frequent token ratio $\alpha$, and a filter cost function $b(\varepsilon)$
\Ensure A chosen data structure (CSF or CSF+Filter) built on the full dataset

\State Compute the optimal $\varepsilon$ by maximizing the equation
\[
\varepsilon^* = \arg\max_{\varepsilon}
\bigl(\alpha \delta (1 - \varepsilon) - (1 - \alpha)b(\varepsilon)\bigr)
\]

\If{$\alpha \delta (1 - \varepsilon^*) - (1 - \alpha)b(\varepsilon^*) - n / N > 0$}
    \State Build a CSF+Filter using the optimal false positive rate $\varepsilon^\star$
\Else
    \State Build a CSF without a filter
\EndIf

\end{algorithmic}
\label{alg:autocsf}
\end{algorithm} 

\subsection{Proving the Lower Bound}

We will now proceed with proving our lower bound on filter-augmented CSF space savings. Our proof of this bound critically leverages the optimality of Huffman codes. 

\begin{theorem}[\citet{huffman1952method}]
\label{thm:HuffmanOptimality}
Given frequencies $F = \{f_1, f_2, ... f_{z}\}$ of values $V = \{v_1, v_2, ... v_{z}\}$ and a Huffman code $H(F)$ with lengths $\{l_1, l_2, ... l_{z}\}$,
$$ \sum_{i = 1}^{z} f_i l_i \leq \sum_{i = 1}^{z} f_i l'_i$$
for any other uniquely decodable code with lengths $l'_i$.
\end{theorem}

Using this foundational result on the optimality of Huffman codes, we now state and prove our new lower bound theorem, which we crucially leverage in the AutoCSF algorithm. 

\begin{theorem}
\label{thm:lower_bound}
Let $C_{\mathrm{FSF}}$ and $C_{CSF}$ denote the memory costs of filtered and non-filtered CSFs, where the filter is a set membership oracle with false positive rate $\epsilon$ that requires $b(\epsilon)$ bits per key.
$$ \left(C_{\mathrm{CSF}} - \mathbb{E}[C_{\mathrm{FSF}}]\right) / N \geq \alpha \delta (1 - \epsilon) - (1 - \alpha)b(\epsilon) - n / N$$
Note that, for practical distributions, the vocabulary size $n$ is typically $o(N)$, making $n / N = o(1)$.
\end{theorem}

\begin{proof}

We suppose that there are $n$ unique values, sorted in order of decreasing frequency $V = \{v_0, ... v_{n}\}$. We use $e$ to denote the Huffman encoding, $l_i$ for the length of the encoding $e(v_i)$ and $f_i$ for the frequency of $v_i$.

\textbf{Cost of the CSF:} The cost to store a compressed static function of $n = |V|$ distinct values is:

$$ C_{\mathrm{CSF}} = \delta \sum_{i = 0}^{n} f_i l_i + D_V + D_e$$

where $\delta$ is a constant that depends on whether $3$ or $4$ hash functions are used to construct the linear system. With $3$ hashes, $\delta_3 \approx 1.089$ and with $4$ hashes, $\delta_4 \approx 1.024$. The canonical Huffman dictionary requires $D_V$ space to store the unique values and $D_e = n\lceil \log_2 \max_i l_i\rceil $ bits to store the encoding lengths for the set of values.

\textbf{Cost of the filtered CSF:} When we prefilter the CSF, we remove (some of) the majority keys from the CSF and change the frequency distribution for the Huffman encoding. Specifically, we reduce the frequency of the majority value from $f_0$ to $f'_0$, causing the Huffman encoding to change from $e$ to $e'$.
Removal of the majority value can potentially update all of the Huffman code lengths, causing them to change from $\{l_0, ... l_{n}\}$ to $\{l'_0, ... l'_{n}\}$. The cost $D_V$ to store the vocabulary remains the same.


The cost of the prefiltered CSF also includes the cost to store the filter, which is $b(\epsilon)$ multiplied by the number of keys in the filter.

$$ C_{\mathrm{FSF}} = \delta f'_0 l'_0 + \delta \sum_{i = 1}^{n} f_i l'_i + D_V + D_{e'} + b(\epsilon) (N - f_0)$$

Here, $f'_0$ denotes the number of keys with value $v_0$ that erroneously pass through the filter. We insert $(N - f_0)$ minority keys into the filter at a cost of $b(\epsilon)$ bits per key, so the filter overhead is $b(\epsilon) (N - f_0)$.

\textbf{Bounding the expected cost:} We wish to find a lower bound on the performance improvement offered by the filter. Note that because $C_{\mathrm{FSF}}$ measures the storage cost (entropy), we want $C_{\mathrm{FSF}} < C_{\mathrm{CSF}}$. We will analyze the following expression.
\begin{align*}
C_{\mathrm{CSF}} - C_{\mathrm{FSF}} &= \delta \sum_{i = 0}^{n} f_i l_i + D_V + D_e \\&- \left(\delta f'_0 l'_0 + \delta \sum_{i = 1}^{n} f_i l'_i + D_V + D_{e'} + b(\epsilon) (N - f_0)\right)
\end{align*}

We may reorganize this expression into the following form.
\begin{align*}
C_{\mathrm{CSF}} - C_{\mathrm{FSF}} &= \delta \left(f_0 l_0 + \sum_{i = 1}^n f_i l_i - f'_0 l'_0 - \sum_{i=1}^n f_i l'_i\right) \\
& + D_e - D_{e'} \\
& + D_V - D_V \\ 
& - b(\epsilon) (N - f_0) 
\end{align*}

We will bound each of these terms individually.

\textbf{Bounding the Huffman code entropy:} To bound the terms involving $f_i l_i$ and $f_i l'_i$, we use the optimality of Huffman codes. Specifically, we have that
$$ f'_0 l'_0 + \sum_{i = 1}^{n} f_i l'_i \leq f'_0 l_0 + \sum_{i=1}^{n} f_i l_i$$
because the code $e'$ designed for $\{f'_0, f_1, f_2, ... f_{n}\}$ will have smaller average space than the code $e$ designed for $F = \{f_0, f_1, f_2, ... f_{n}\}$ (by Theorem~\ref{thm:HuffmanOptimality}). Therefore,

\begin{align*}
\delta \left(f_0 l_0 + \sum_{i = 1}^n f_i l_i - f'_0 l'_0 - \sum_{i=1}^n f_i l'_i\right)&\\
 \geq \left(f_0 l_0 + \sum_{i = 1}^n f_i l_i - f'_0 l_0 - \sum_{i=1}^n f_i l_i\right)& = l_0 (f_0 - f'_0)
\end{align*}

\textbf{Bounding the Huffman codebook size:} To bound the size of the difference $D_{e'} - D_{e}$, we rely on some basic properties of canonical Huffman codes.
Canonical Huffman codes only need to store the lengths of each code in $e(V)$ and $e'(V)$, not the codes themselves. Let $l_{\max} = \max_i l_i$ and $l'_{\max} = \max_i l'_i$. 

Then
$$D_e = n\lceil \log_2 l_{\max} \rceil \qquad D_{e'} = n\lceil \log_2 l'_{\max}\rceil$$

By replacing $f_0$ with $f'_0$, we increase the depth of the Huffman tree for $e'$ by at most 1. Therefore, $l'_{\max} \leq l_{\max} + 1$. This implies the following bound on the canonical Huffman codebooks.

$$ D_{e'} - D_{e} \leq n \lceil \log_2 (l_{\max}+1) \rceil - z\lceil \log_2 l_{\max} \rceil \leq n$$

\textbf{Taking the expectation:} By combining the previous two parts of the proof, we have the bound 

\begin{align*}
C_{\mathrm{CSF}} - C_{\mathrm{FSF}} &\geq \delta l_0 (f_0 - f'_0)  -n  - b(\epsilon) (N - f_0) 
\end{align*}

The headroom expression $C_{\mathrm{CSF}} - C_{\mathrm{FSF}}$ contains $f'_0$, the number of keys with value $v_0$ that erroneously pass through the filter. It should be noted that $f'_0$ is a random variable, where the randomness comes from the hash functions and other randomization techniques used in the approximate set membership algorithm. 

Because the set membership algorithm returns a false positive with probability $\epsilon$, the expected value of $f'_0$ is $\mathbb{E}[f'_0] = \epsilon f_0$. Thus we have

\begin{align*}
C_{\mathrm{CSF}} - \mathbb{E}[C_{\mathrm{FSF}}] &\geq \delta l_0 f_0 (1 - \epsilon)
-n -b(\epsilon) (N - f_0) 
\end{align*}

To complete the proof, we divide both sides of the equation by $N$, the total number of keys in the key-value store.
\begin{align*}
\frac{1}{N}\left(C_{\mathrm{CSF}} - \mathbb{E}[C_{\mathrm{FSF}}]\right) \geq \delta l_0 \alpha (1 - \epsilon) 
-n / N 
-b(\epsilon) (1 - \alpha) 
\end{align*}

\end{proof}

\subsection{Proving the Upper Bound}
Now we develop an upper bound on the memory cost of the filtered CSF relative to the regular CSF. In order to do this, we will leverage an additional result from information theory. The following theorem is a well-known consequence of the Kraft-McMillan inequality \cite{kraft1949device, mcmillan1956two} and the Shannon noiseless coding theorem \cite{shannon1948mathematical, cover1999elements}.

\begin{theorem}\label{thm:HuffmanEntropy} Given frequencies $F = \{f_1, f_2, ... f_z\}$ of values $V = \{v_1, v_2, ... v_z\}$ and a Huffman code $H(F)$ with lengths $\{l_1, l_2, ... l_z\}$

$$ -\sum_{i = 0}^{z}\frac{f_i}{N} \log_2 \frac{f_i}{N} \leq \sum_{i = 0}^{z} \frac{f_i}{N} l_i \leq 1 - \sum_{i = 0}^{z}\frac{f_i}{N} \log_2 \frac{f_i}{N} $$

where $N = \sum_{i = 1}^{z} f_i$.
\end{theorem}

Now, we state and prove our new upper bound result. 

\begin{theorem}
\label{thm:upper-bound}
Let $C_{\mathrm{FSF}}$ and $C_{CSF}$ denote the memory costs of filtered and non-filtered CSFs, where the filter is a set membership oracle with false positive rate $\epsilon$ that requires $b(\epsilon)$ bits per key.
$$\mathbb{E}[C_{\mathrm{CSF}}- C_{\mathrm{FSF}}] / N  \leq 2\delta - \frac{1}{2}\alpha  \delta  \epsilon + n/N 
- b(\epsilon) (N - f_0)$$
\end{theorem}
\begin{proof}
We follow the same development as the previous theorem, up to the point where we obtain the following cost expression:

\begin{align*}
C_{\mathrm{CSF}} - C_{\mathrm{FSF}} &= \delta \left(f_0 l_0 + \sum_{i = 1}^n f_i l_i - f'_0 l'_0 - \sum_{i=1}^n f_i l'_i\right) \\
& + D_e - D_{e'} \\
& + D_V - D_V \\ 
& - b(\epsilon) (N - f_0) 
\end{align*}

Rather than find lower bounds for each term as before, now we will find upper bounds. By replacing $f_0$ with $f_0'$, we may decrease the depth of the Huffman tree by at most 1, yielding $D_e - D_{e'} \leq n$. As before, we leave the $-b(\epsilon)(N - f_0)$ term alone, and turn our attention to the more difficult term involving the code lengths.

The main idea behind this proof is to use both the upper and lower bounds from Theorem~\ref{thm:HuffmanEntropy} to get an upper bound on the difference. Therefore, we use the upper entropy bound for the $f_i l_i$ terms and the lower entropy bound for the $-f_i' l_i'$ terms.

$$- \sum_{i=0}^{n}\frac{f_i}{N} \log_2\frac{f_i}{N} \leq \frac{f_0}{N}l_0 + \sum_{i=1}^{n}\frac{f_i}{N}l_i \leq 1 - \sum_{i=0}^{n}\frac{f_i}{N}\log_2 \frac{f_i}{N}$$

This yields the following pair of inequalities.


\begin{align*}
\frac{f_0'}{N} l_0' + \sum_{i=1}^{n}\frac{f_i}{N}l_i' &\geq -\frac{f_0'}{N'} \log_2\frac{f_i'}{N'} - \sum_{i=1}^{n} \frac{f_i}{N'} \log_2 \frac{f_i}{N'}\\
\frac{f_0'}{N} l_0' + \sum_{i=1}^{n}\frac{f_i}{N}l_i' &\leq 1 - \frac{f_0'}{N'} \log_2 \frac{f_0'}{N'} - \sum_{i=1}^{n} \frac{f_i}{N'} \log_2 \frac{f_i}{N'}
\end{align*}

Note that $N' = N - f_0 + f_0'$ because the two Huffman codes are constructed on different sets: one on the original data and the second on the filtered data.

This allows us to bound the following expression:
$$Z = f_0 l_0 + \sum_{i=1}^{n} f_i l_i - f_0' l_0' - \sum_{i=1}^{n} f_i l_i'$$
$$ \leq N - f_0 \log_2 \frac{f_0}{N} + f_0' \log_2 \frac{f_0'}{N'} - \sum_{i=1}^{n}f_i\left(\log_2 \frac{f_i}{N} - \log_2 \frac{f_i}{N'}\right)$$
$$ = N - f_0 \log_2 \frac{f_0}{N} + f_0' \log_2 \frac{f_0'}{N'} - \sum_{i=1}^{n}f_i\left(\log_2 N' - \log_2 N\right)$$
$$ = N - f_0 \log_2 \frac{f_0}{N} + f_0' \log_2 \frac{f_0'}{N'} - \log_2\left(\frac{N'}{N}\right)\sum_{i=1}^{n}f_i$$

For the $f_0' \log_2 f_0' / N'$ term, note that
\begin{align*}
f_0' \log_2 \frac{f_0'}{N'} &= f_0' \log_2 \left(\frac{f_0'}{N-f_0 + f_0'}\right) = f_0' \log_2 \left(\frac{1}{\frac{N-f_0}{f_0'} + 1}\right)\\ &= -f_0' \log_2\left(1 + \frac{N-f_0}{f_0'}\right)
\end{align*}

Recall the inequality
$-\ln(1 + z) \leq -\frac{2z}{2+z}$ for $z \geq 0$~\cite{topsoe2004some}.
Applying this inequality to our expression with $z = (N-f_0)/f_0'$, we obtain

\begin{align*}
-f_0'\log_2\left(1 + \frac{N-f_0}{f_0'}\right) &\leq -2f_0'\frac{N-f_0}{f_0'\left(2 + \frac{N-f_0}{f_0'}\right)} = -2f_0'\frac{N-f_0}{2f_0' + N-f_0}\\ &\leq -2f_0'\frac{N-f_0}{3N - f_0}
\end{align*}
where the second inequality is because $f_0' \leq N$. This leaves us with the overall inequality

$$Z \leq N - f_0 \log_2 \frac{f_0}{N} - 2f_0' \frac{N-f_0}{3N-f_0} - \log_2\left(\frac{N'}{N}\right) \sum_{i=1}^{n}f_i$$

Now we address the $-\log_2\left(\frac{N'}{N}\right) \sum_{i=1}^{n}f_i$ term, observing that the function
$$-\log_2\left(\frac{N'}{N}\right) = -\log_2\left(\frac{N-f_0 + f_0'}{N}\right) \leq -\log_2\left(\frac{N-f_0}{N}\right)$$
because it is a monotonically decreasing function of $f_0'$ that attains a maximum at $f_0' = 0$. Finally, note that
$\sum_{i=1}^{n}f_i = N - f_0$ by the definition of $N$,
yielding the following overall bounds.

$$Z \leq N - f_0 \log_2 \frac{f_0}{N} - 2f_0' \frac{N-f_0}{3N-f_0} - \log_2\left(\frac{N-f_0}{N}\right) \left(N-f_0\right)$$

We may reorganize this expression into the following form by reordering terms and using the fact that $\alpha = f_0 / N$.

$$Z \leq N - f_0 \log_2 \alpha - 2f_0' \frac{1-\alpha}{3-\alpha} - \log_2\left(1-\alpha\right) \left(N-f_0\right)$$

By combining this bound with the previous bound on $D_e - D_{e'}$ we have the following bound on $\Delta = C_{\mathrm{CSF}} - C_{\mathrm{FSF}}$

\begin{align*}
\Delta &= \delta \left(N - f_0 \log_2 \alpha - 2f_0' \frac{1-\alpha}{3-\alpha} - (N-f_0)\log_2 (1-\alpha)\right) \\
& + n - b(\epsilon) (N - f_0) 
\end{align*}

Taking the expectation and dividing by $N$, we obtain
\begin{align*}
\mathbb{E}[\Delta] / N &\leq \delta \left(1 - \alpha \log_2 \alpha - 2\alpha \epsilon \frac{1-\alpha}{3-\alpha} - (1-\alpha)\log_2 (1-\alpha)\right) \\
& + n/N - b(\epsilon) (N - f_0)
\end{align*}
Observe that $-\alpha \log_2 \alpha - (1-\alpha) \log_2(1-\alpha)$ is the entropy of a Bernoulli trial with success rate $\alpha$, which is trivially bounded above by 1 bit. This gives rise to the expression
$$\mathbb{E}[\Delta] / N  \leq 2\delta - 2\alpha  \delta \epsilon \frac{1-\alpha}{3-\alpha} + n/N 
- b(\epsilon) (1 - \alpha)$$
which can be further simplified to
$$\mathbb{E}[\Delta] / N  \leq 2\delta - \frac{1}{2}\alpha \delta \epsilon + n/N 
- b(\epsilon) (1 - \alpha)$$
\end{proof}

Putting the two bounds together,

\begin{align*}
\mathbb{E}[C_{\mathrm{CSF}}- C_{\mathrm{FSF}}] / N &\geq \alpha \delta (1-\epsilon) - (1-\alpha) b(\epsilon) - n/N \\
\mathbb{E}[C_{\mathrm{CSF}}- C_{\mathrm{FSF}}] / N &\leq 2\delta - \frac{1}{2}\alpha \epsilon
- b(\epsilon) (1 - \alpha) + n/N 
\end{align*}

The gap is
$ 2\delta + \frac{1}{2}\alpha \delta \epsilon - \alpha \delta$
which is greater than 0.

\section{Theory Validation}

\begin{figure*}[!t]
\begin{center}
\includegraphics[width=\textwidth]{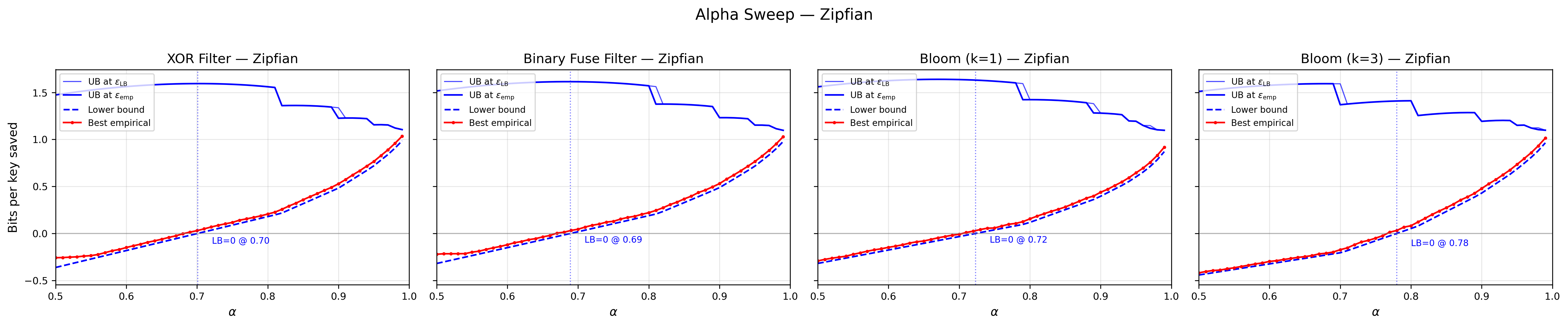}
\end{center}

\begin{center}
\includegraphics[width=\textwidth]{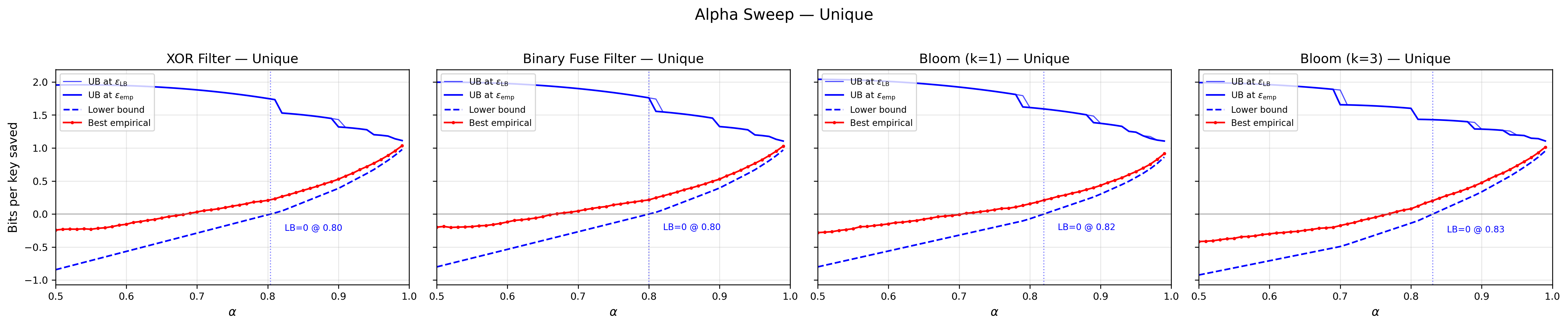}
\end{center}

\begin{center}
\includegraphics[width=\textwidth]{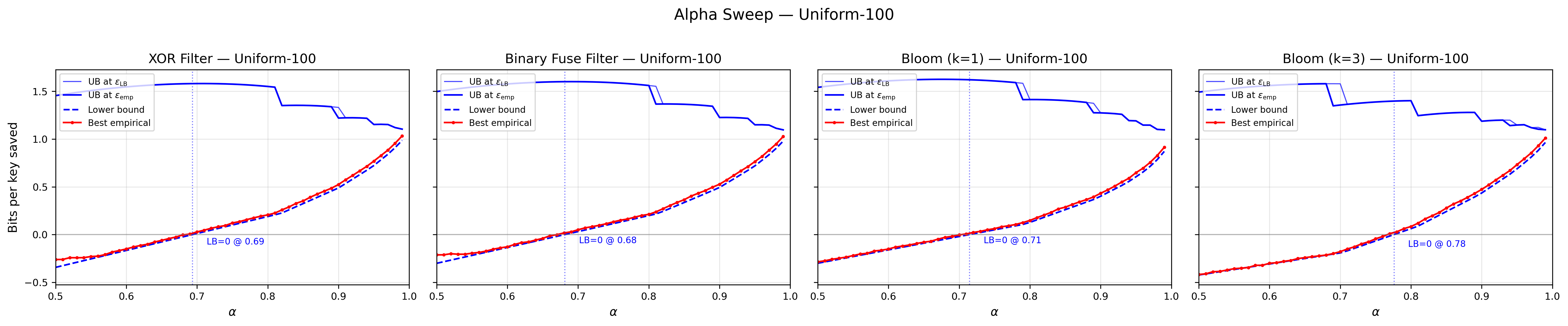}
\end{center}

\caption{Results of sweeping $\alpha$ for each distribution across all four filter types. Each panel plots the lower bound (blue dashed), best empirical savings (red dashed), and the upper bound at the theory-guided and empirical-best parameters (light and dark blue solid). The vertical dashed line marks where the lower bound crosses zero.}
\label{fig:alpha_sweep}
\end{figure*}

The AutoCSF algorithm is based on three theoretical claims:
\begin{enumerate}
    \item The space savings from filter augmentation are bounded between our lower and upper bounds (Theorems 5.2 and 5.4).
    \item The sign of the lower bound provides a reliable decision criterion for when filtering is beneficial.
    \item Maximizing the lower bound over the filter's parameter space yields the parameter with the largest provable space savings.
\end{enumerate}
In this section, we evaluate each of these claims empirically. Our goal is to determine whether the bounds are tight enough to be practically useful, whether the decision criterion produces false positives or false negatives, and how much space the theory-guided parameter sacrifices, if any, relative to an exhaustive search over all discrete parameter settings.

We construct synthetic key-value datasets that sweep the majority fraction $\alpha$ from 0.50 to 0.99 and evaluate three minority value distributions that span a range of minority entropy.
\begin{itemize}
    \item \textbf{Unique}: each minority key has a distinct value (worst case for the bounds).
    \item \textbf{Zipfian}: power-law with exponent $s = 1.5$.
    \item \textbf{Uniform-100}: values drawn uniformly from 100 symbols.
\end{itemize}
We evaluated four types of filters: XOR filter, binary fuse filter, and Bloom filters with $k \in \{1, 3\}$ hash functions.

Figure~\ref{fig:alpha_sweep} shows the results of sweeping $\alpha$ for each distribution across all four filter types. Each panel plots four curves: the lower bound (blue dashed), best empirical savings (red dashed), and the upper bound evaluated at the theory-guided parameter (light blue solid) and at the empirical-best parameter (dark blue solid). A vertical dashed line marks where the lower bound crosses zero.

\begin{figure*}[p]
\makeatletter\setlength{\@fptop}{0pt}\makeatother
\vspace*{-2cm} 
\centering
\includegraphics[width=0.8\textwidth]{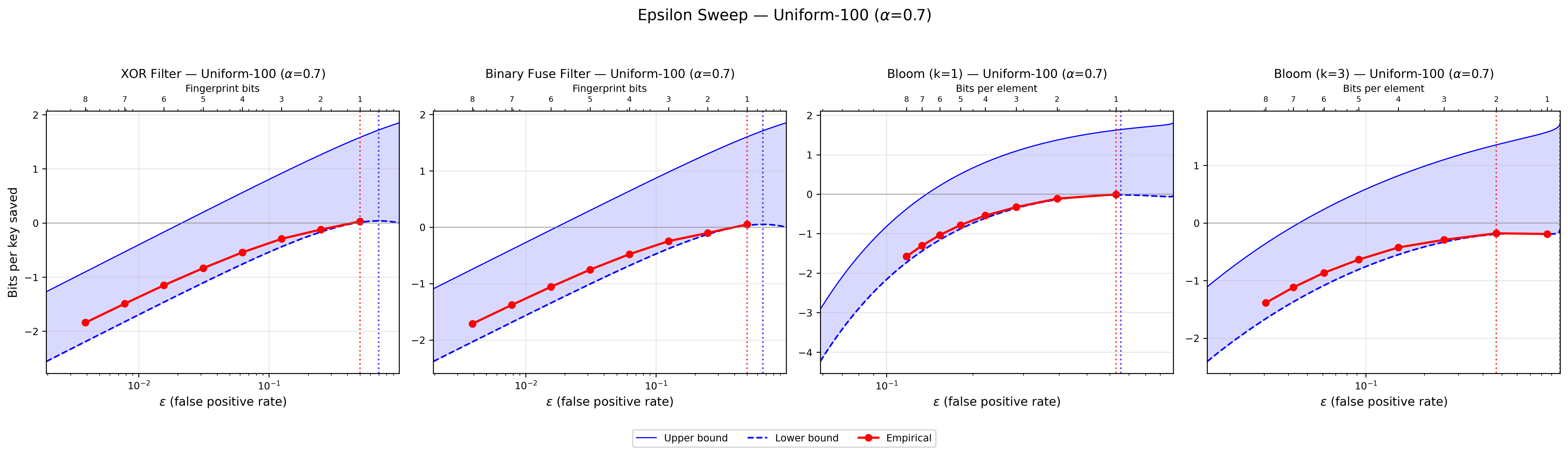}\\[-2pt]
\includegraphics[width=0.8\textwidth]{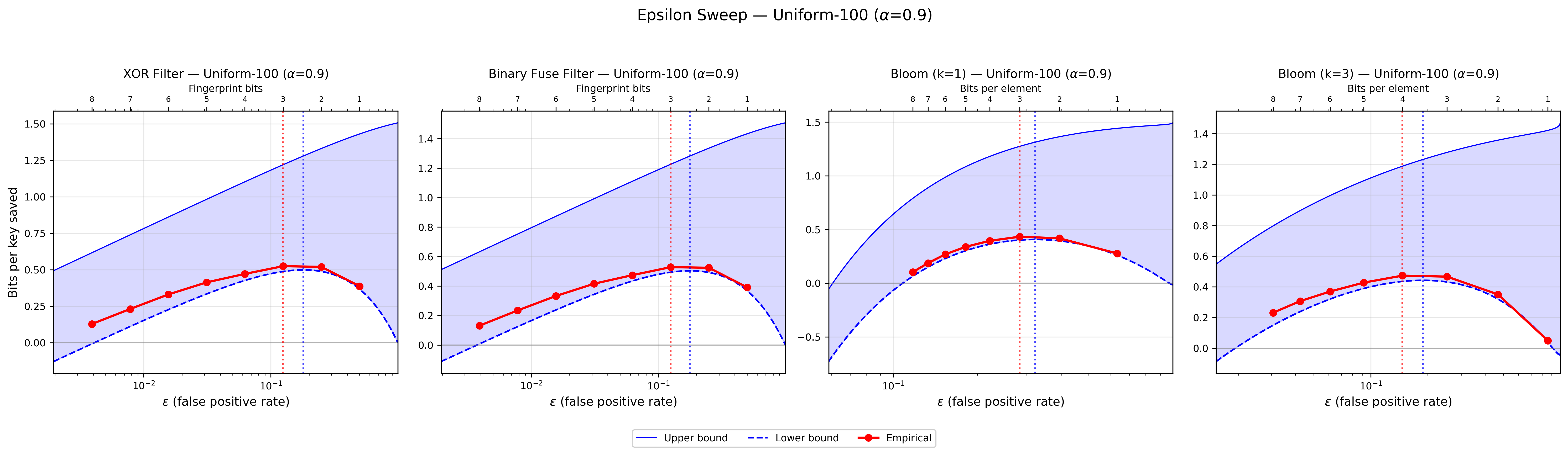}\\[4pt]
\includegraphics[width=0.8\textwidth]{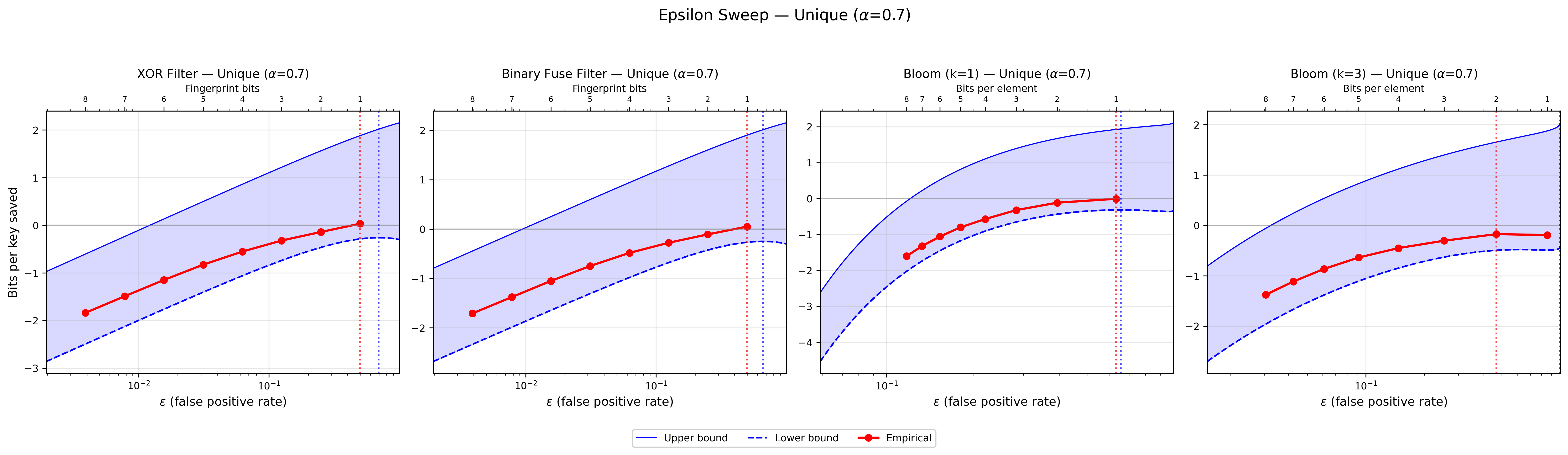}\\[-2pt]
\includegraphics[width=0.8\textwidth]{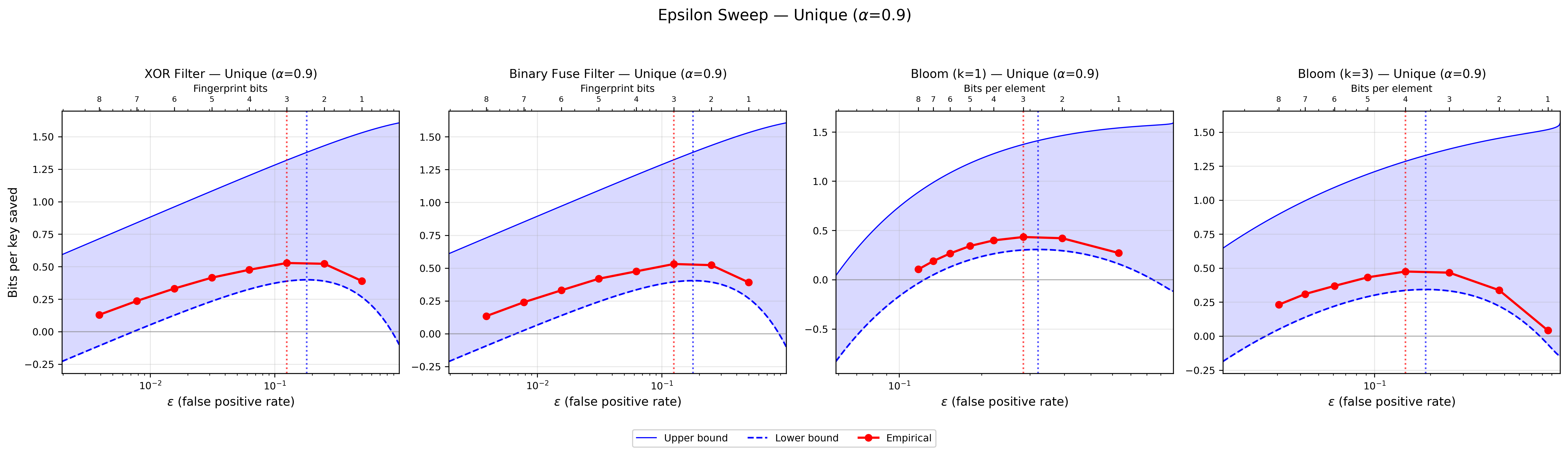}\\[4pt]
\includegraphics[width=0.8\textwidth]{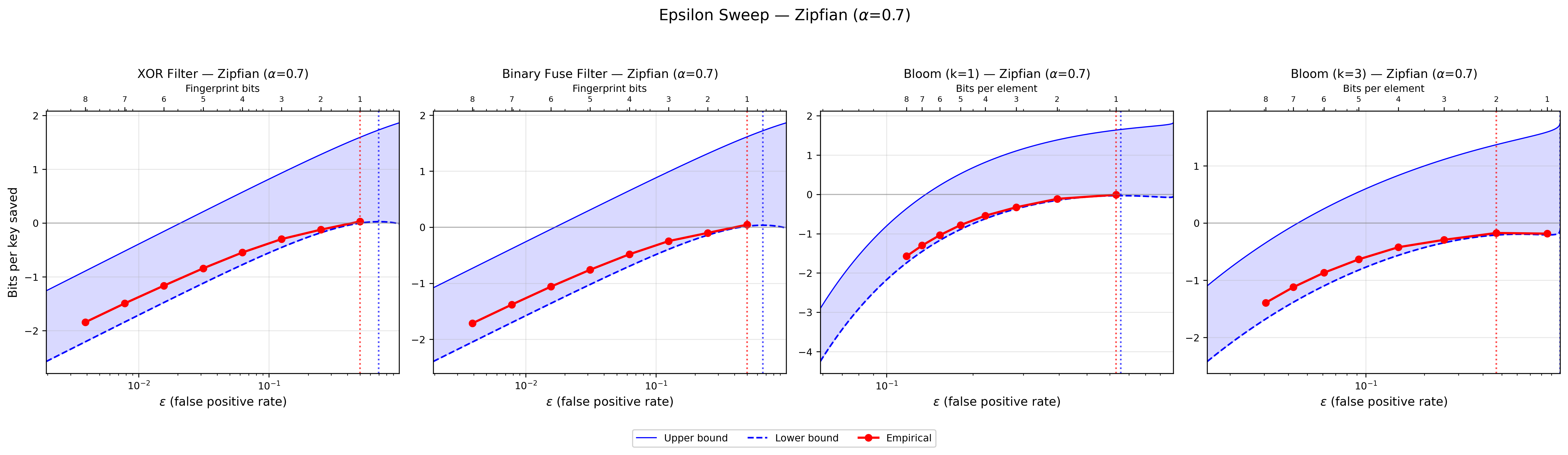}\\[-2pt]
\includegraphics[width=0.8\textwidth]{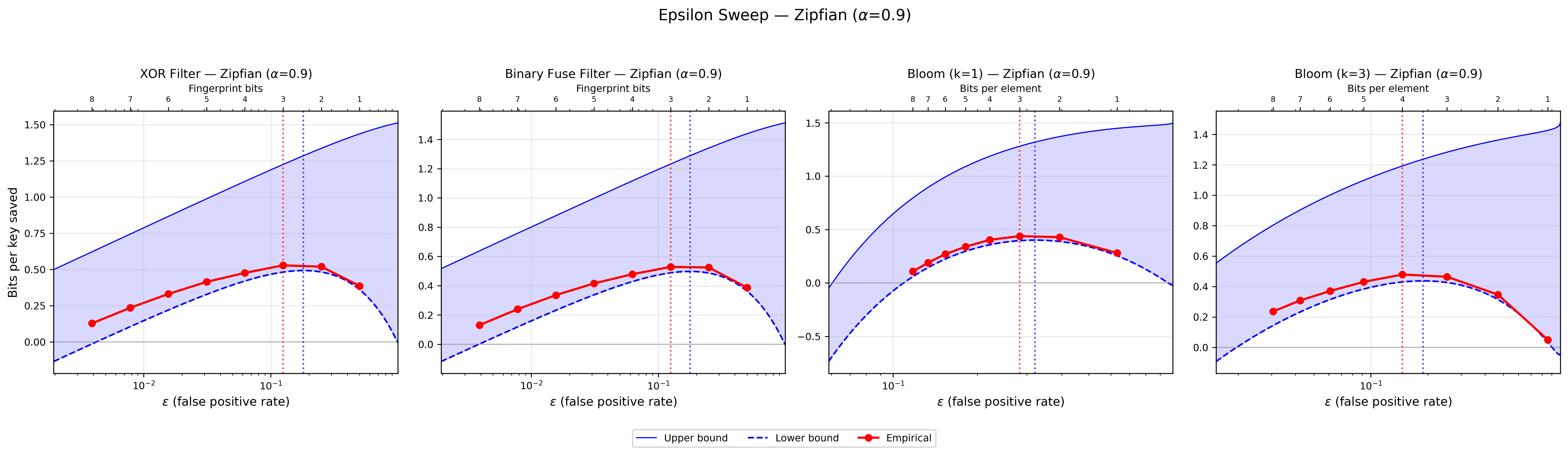}
\caption{Results of sweeping $\varepsilon$ for each distribution across all four filter types.}
\label{fig:epsilon_sweep}
\end{figure*}

\subsection{Bound Tightness}

Figure~\ref{fig:alpha_sweep} shows four curves per panel. The \emph{lower bound} (blue dashed) is Theorem~5.2 evaluated at the discrete parameter $\varepsilon_{\mathrm{LB}}$ that maximizes the lower bound. The \emph{best empirical} curve (red dashed) is the measured bits/key saved at the discrete parameter $\varepsilon_{\mathrm{emp}}$ that yields the largest savings, found by exhaustive search over discrete filter parameters. The \emph{UB at $\varepsilon_{\mathrm{LB}}$} curve (light blue solid) is the upper bound (Theorem~5.4) evaluated at the theory-guided parameter, and the \emph{UB at $\varepsilon_{\mathrm{emp}}$} curve (dark blue solid) is the upper bound at the empirical-best parameter. A vertical dashed line marks the $\alpha$ value where the lower bound crosses zero, indicating the decision boundary for filter augmentation.

We observe that across all filter-distribution combinations and all 50 values of $\alpha$, the empirical savings consistently fall between the lower bound and the upper bound at $\varepsilon_{\mathrm{emp}}$. We note that this is not a formal proof of correctness. The bounds assume a continuous $\varepsilon$ while the experiments operate over discrete filter parameters, and the measured space includes implementation overheads (metadata, alignment) not captured by the theory. Nonetheless, empirical values stay within the predicted range across all configurations, which provides strong evidence that the bounds are practically valid.

The upper bound is looser, with a gap of up to $\sim$1.5 bits/key at moderate $\alpha$ for the unique distribution, and a smaller gap for zipfian and uniform-100. This looseness does not affect AutoCSF's practical decisions, which are driven entirely by the lower bound. The value of the upper bound is primarily theoretical in establishing, to our knowledge, the first such non-trivial guarantee on the total space savings. This upper bound establishes the current best-known result. A promising direction for future work might be to tighten this bound further, perhaps by making more fine-grained assumptions on the value distribution. 

\subsection{Validating the Decision Criterion}

To further validate AutoCSF, we sweep the false positive rate $\varepsilon$ across all discrete filter configurations for each distribution at two values of $\alpha$: $\alpha = 0.7$, where filtering is marginal, and $\alpha = 0.9$, where filtering provides clear benefits. Each panel in Figure~\ref{fig:epsilon_sweep} sweeps $\varepsilon$ on a log scale and plots the continuous lower and upper bounds (shaded region). The empirical bits/key saved are overlaid at the discrete $\varepsilon$ values corresponding to each filter parameter (e.g., fingerprint bits or bits per element), and the theory-optimal $\varepsilon^*$ is marked as the continuous maximum of the lower bound.

\textbf{Decision criterion.} The lower bound crossing zero provides a decision boundary: when $\text{LB} > 0$ at the theory-optimal $\varepsilon^*$, AutoCSF recommends filtering; otherwise, it defaults to a plain CSF. Across all filter types, distributions, and $\alpha$ values in the alpha sweep, we find that whenever the lower bound recommends filtering, the empirical savings are positive, i.e. the criterion never recommends a filter that increases space. The criterion is conservative, however, for the unique distribution: there is a range of $\alpha$ values (roughly $0.71$--$0.80$ for XOR) where filtering produces small empirical savings (${\sim}0.05$--$0.2$ bits/key) but the lower bound remains negative. This gap is driven by the $n/N$ term, which is large for the unique distribution where $n \approx N(1-\alpha)$. For zipfian and uniform-100, where $n/N \approx 0$, the lower bound crossing aligns closely with the point at which filtering first becomes empirically helpful.

\textbf{Parameter selection.} In practice, each filter type admits only a small set of discrete configurations (e.g., fingerprint\_bits $\in \{1, \ldots, 8\}$ for XOR filters, or all $(k, \text{bpe})$ pairs for Bloom filters), each corresponding to a specific $\varepsilon$. AutoCSF evaluates the lower bound at each available $\varepsilon$ and selects the configuration whose lower bound is largest. This search is inexpensive since the parameter space is small, and it avoids any mismatch between a continuous optimum and the discrete $\varepsilon$ values that filters can actually realize. To evaluate how well the lower bound serves as a proxy for actual savings, we compare the parameter selected by the lower bound against the empirically optimal parameter found by exhaustive search. The two agree in the large majority of cases, and when they disagree it is always by a single discrete step, with the theory-guided parameter saving at most $0.02$ bits/key less than the empirical optimum.

\subsection{Comparison with Prior Work}

The closest prior work to AutoCSF is the Bloom-enhanced CSF (BCSF) of Shibuya et al.~\cite{shibuya2022space}. Like AutoCSF, their method augments a CSF with a Bloom filter to handle the majority value. Their approach differs from ours in three respects:
\begin{itemize}
\item \textbf{CSF cost model.} They derive the optimal false positive rate $\varepsilon^*$ using a heuristic cost model $C_{\mathrm{CSF}}(H_0)$ fit empirically to observed CSF sizes, rather than from provable bounds.
\item \textbf{Bloom filter cost model.} Their Bloom filter cost model uses the idealized constant $C_{\mathrm{BF}} = \log_2 e \approx 1.44$ bits per key, which assumes a continuous (real-valued) number of hash functions, rather than a parameterized model that accounts for the discrete choice of hash count $k$ and bits per element.
\item \textbf{Decision criterion.} Their method recommends no filter when $\varepsilon^* \geq 1$, but this criterion is rarely triggered in practice: the heuristic $C_{\mathrm{CSF}}(H_0)$ produces values large enough that the method recommends filtering even at $\alpha = 0.50$, where filtering does not reduce space.
\end{itemize}

To compare the two approaches, we restrict AutoCSF to Bloom filters only (selecting the best $(k, \text{bpe})$ pair via the lower bound) and measure the resulting bits per key across all three distributions. Figure~\ref{fig:shibuya_bpk} shows the measured bits per key for each method alongside a no-filter baseline.

The comparison reveals a clear difference in the \emph{decision} to filter. At low $\alpha$, the heuristic recommends filter configurations that \emph{increase} space usage relative to a plain CSF, with the heuristic curve exceeding the no-filter baseline by up to 1 bit/key for the unique distribution at $\alpha = 0.50$. This gap persists across all three distributions up to $\alpha \approx 0.7$--$0.8$ depending on the minority entropy. AutoCSF avoids this by defaulting to a plain CSF whenever the lower bound is negative. When filtering \emph{is} beneficial (roughly $\alpha \gtrsim 0.8$), both methods achieve comparable space usage, with AutoCSF matching or improving upon the heuristic. These results confirm that AutoCSF improves upon the heuristic both in the decision of when to apply a filter and in the selection of filter parameters.

We restrict both methods to Bloom filters in this comparison for fairness, since the BCSF framework only supports Bloom filters. However, as shown in the alpha sweep experiments, AutoCSF can also leverage XOR filters, binary fuse filters, and other set membership structures, which consistently yield further space savings.

\begin{figure*}[!t]
\includegraphics[width=0.8\textwidth]{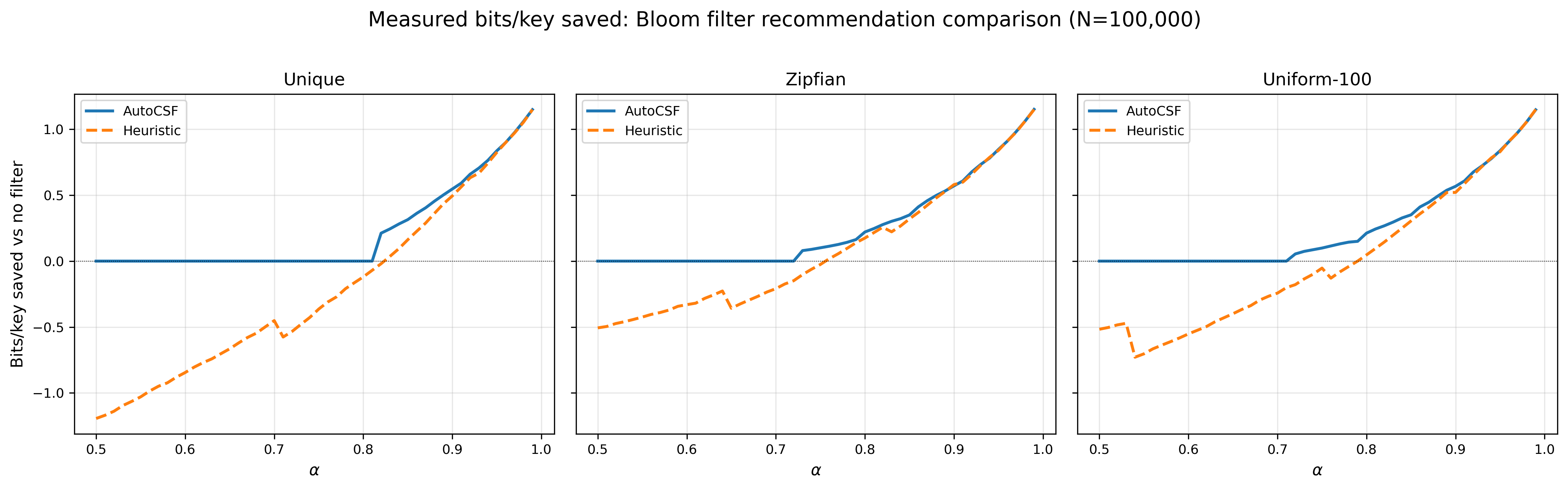}
\caption{Results of comparing heuristic and AutoCSF cost models}
\label{fig:shibuya_bpk}
\end{figure*}

\section{Experiments}

The previous section validates AutoCSF's theoretical bounds and decision criterion in isolation. We now evaluate AutoCSF as a complete, end-to-end index and compare it against four baselines that span the landscape of practical approaches to static key-value indexing. Our goal is to answer three questions: (1)~Does AutoCSF's theory-guided filter selection translate into real space savings over competing methods? (2)~What are the latency and construction costs, if any, of filter-augmented CSFs? (3)~How do these trade-offs change on real-world workloads where the value distributions are not synthetically controlled?

\subsection{Baselines}

We compare AutoCSF against the following methods:

\begin{itemize}
\item \textbf{BCSF (Shibuya)}~\cite{shibuya2022space}. The Bloom-enhanced CSF described in Section~3. We use \citeauthor{shibuya2022space}'s heuristic cost model to set the Bloom filter parameters. This baseline isolates the effect of AutoCSF's improved parameter selection: both methods use a CSF with a pre-filter, but differ in how the filter parameters are chosen.

\item \textbf{C++ Hash Table}. A standard \texttt{std::unordered\_map} storing all key-value pairs. This baseline represents the naive approach with no compression, providing a reference point for memory usage and query latency.

\item \textbf{MPH Table}. A minimal perfect hash function (via Sux4J's GOV construction~\cite{genuzio2020fast}) paired with a compact value array. This is a distribution-agnostic baseline as it does not exploit value skew and its memory usage depends only on the vocabulary size.

\item \textbf{Learned CSF}~\cite{hermann2025learned}. A learned static function that trains a small neural network to predict value distributions from keys. These predictions are used to derive key-specific prefix codes, which are stored in a pair of ribbon-based retrieval structures. By learning correlations between keys and values, this method can break the zero-order entropy barrier of standard CSFs, but at the cost of higher query latency and construction time due to model training.
\end{itemize}

\subsection{Setup}

\textbf{Synthetic datasets.} We generate key-value datasets with $N = 100{,}000$ keys, sweeping the majority fraction $\alpha$ from 0.50 to 0.99. The majority value is assigned to an $\alpha$-fraction of keys uniformly at random; the remaining $(1-\alpha)N$ minority keys draw values from one of three distributions: \emph{Uniform-100} (100 equally likely symbols), \emph{Zipfian} (power-law with exponent $s = 1.5$), and \emph{Unique} (every minority key has a distinct value). These distributions span a wide range of minority entropy and stress-test different aspects of each method's compression strategy.

\textbf{Genomics datasets.} We evaluate on the three real $k$-mer count datasets used in \cite{shibuya2022space} derived from whole-genome sequencing data, using $k = 15$:
\begin{itemize}
\item \emph{E.\ coli} (Sakai strain): $n = 5.3$M keys, $\alpha = 0.97$, 42 distinct values.
\item \emph{SRR10211353}: $n = 9.8$M keys, $\alpha = 0.20$, 228 distinct values.
\item \emph{C.\ elegans}: $n = 69.7$M keys, $\alpha = 0.82$, 760 distinct values.
\end{itemize}
These datasets exhibit the diversity of real genomics workloads: E.\ coli has extreme skew ($\alpha = 0.97$) with very few distinct values, SRR has low skew ($\alpha = 0.20$) with moderate vocabulary, and C.\ elegans is a large-scale dataset (70M keys) with intermediate skew.

\textbf{Metrics.} We report three metrics: \emph{memory} in bits per key (bpk), \emph{query latency} as the mean inference time per key in nanoseconds, and \emph{construction time} in seconds. For AutoCSF and BCSF, we use binary fuse filters~\cite{graf2020xor}, which offer lower false positive rates per bit than Bloom filters.

\textbf{Learned CSF configuration.} We follow the training pipeline of \citet{hermann2025learned}, sweeping four MLP architectures (linear, one hidden layer with 50 or 100 units, and two hidden layers with 50 units each) with float16 quantization and early stopping. For genomics datasets, we use ordinal $k$-mer encoding as input features, mapping each nucleotide position to $\{0,1,2,3\}/3$, since the character-level structure of $k$-mers is learnable unlike the hash-derived features used for synthetic data. The best model is selected by minimum total size (ribbon storage plus model weights). Construction time includes both model training and index construction.

\subsection{Synthetic Benchmarks}

\textbf{Pareto frontier.} Figure~\ref{fig:pareto} plots memory (bpk) against query latency (ns) on log-log axes for all five methods across the three synthetic distributions plus genomics data. Each method traces a trajectory as $\alpha$ sweeps from 0.5 to 0.99; larger markers highlight $\alpha \in \{0.5, 0.8, 0.95\}$.

\begin{figure*}[t]
\centering
\includegraphics[width=\textwidth]{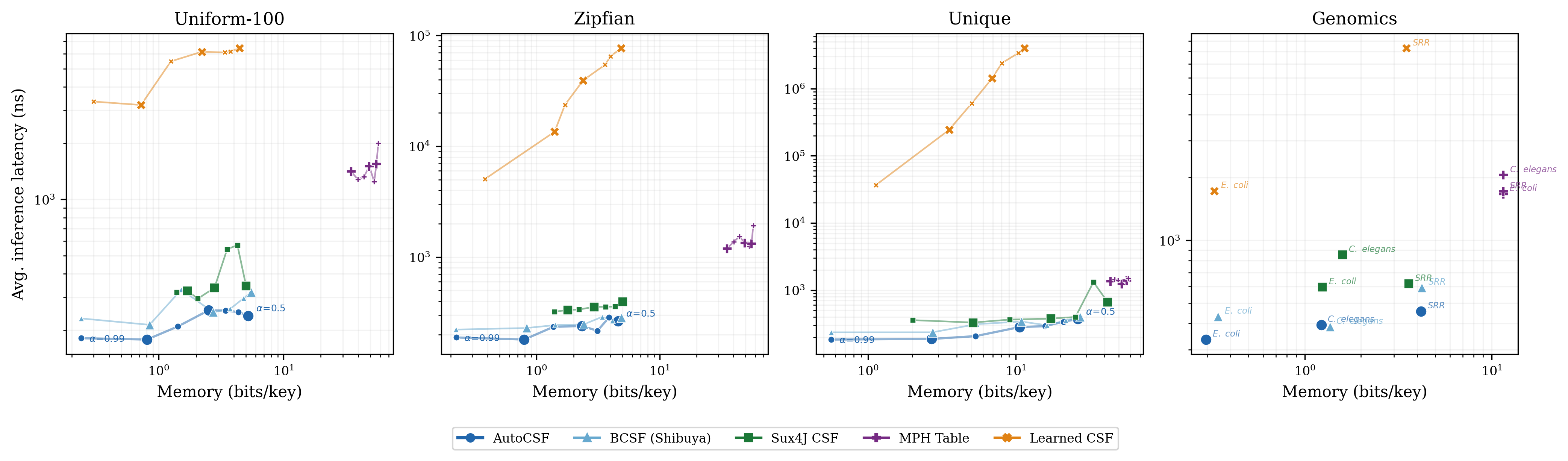}
\caption{Memory--latency Pareto frontier across synthetic distributions and genomics datasets ($N = 100{,}000$ for synthetic, real sizes for genomics). Each trajectory sweeps $\alpha$ from 0.5 (right, more memory) to 0.99 (left, less memory). Log-log scale.}
\label{fig:pareto}
\end{figure*}

Several patterns emerge from the Pareto plot. AutoCSF and BCSF consistently occupy the lower-left corner of the Pareto frontier---low memory \emph{and} low latency---across all three distributions. The hash table offers competitive latency (107--204~ns) but does not compress the data. MPH Table uses 7--70$\times$ more memory than AutoCSF and its latency is 10--20$\times$ higher. Learned CSF achieves tighter memory than AutoCSF on Uniform-100 at all $\alpha$ values (e.g., 4.4 vs.\ 5.2~bpk at $\alpha = 0.5$) and on Unique at low-to-moderate $\alpha$ (e.g., 11.4 vs.\ 26.3~bpk at $\alpha = 0.5$), but at 2--4 orders of magnitude higher query latency. On Zipfian, AutoCSF achieves lower memory than Learned CSF at all $\alpha$ values. The log-log scale in Figure~\ref{fig:pareto} is necessary to display these extremes on the same axes.

The Unique distribution is the hardest case for all methods: at $\alpha = 0.5$, half the keys have distinct values, so the CSF must store nearly $\log_2(N/2) \approx 16$ bits of entropy per minority key. Learned CSF achieves 11.4~bpk versus AutoCSF's 26.3~bpk, demonstrating that a trained model can significantly reduce effective entropy when the value space is large. However, this comes at extreme cost: 4.0M~ns query latency ($34{,}000\times$ slower than AutoCSF's 116~ns) and 983 seconds to construct ($15{,}000\times$ slower). At higher skew ($\alpha = 0.95$), AutoCSF overtakes Learned CSF in memory (2.7 vs.\ 3.5~bpk) while maintaining $3{,}000\times$ faster queries.

\textbf{Construction time.} Table~\ref{tab:synthetic} reports memory, latency, and construction time for all five methods at $\alpha \in \{0.5, 0.8, 0.95\}$. AutoCSF and BCSF build in under 0.1 seconds across all configurations. C++ Hash Table is the fastest to construct (${\sim}0.006$~s). MPH Table builds in 0.4--0.6 seconds. Learned CSF is the most expensive to build: 5--983 seconds depending on distribution complexity, with the Unique distribution requiring over 16 minutes at $\alpha = 0.5$. At $\alpha = 0.95$ on Uniform-100, where Learned CSF's memory advantage over AutoCSF is only 0.1~bpk, it takes $440\times$ longer to construct (5.7~s vs.\ 0.013~s).

\begin{table*}[t]
\centering
\begin{tabular}{lrrrrrrrrr}
\toprule
 & \multicolumn{3}{c}{$\alpha\!=\!0.5$} & \multicolumn{3}{c}{$\alpha\!=\!0.8$} & \multicolumn{3}{c}{$\alpha\!=\!0.95$} \\
Method & bpk & query (ns) & build (s) & bpk & query (ns) & build (s) & bpk & query (ns) & build (s) \\
\midrule
\multicolumn{10}{c}{\emph{Uniform-100}} \\
\midrule
AutoCSF & 5.2 & \textbf{92} & 0.030 & 2.5 & \textbf{116} & 0.030 & 0.8 & \textbf{78} & 0.013 \\
BCSF (Shibuya) & 5.5 & 95 & 0.028 & 2.7 & 120 & 0.020 & 0.8 & 87 & 0.008 \\
C++ Hash Table & 95.1 & 150 & \textbf{0.006} & 95.1 & 121 & \textbf{0.006} & 95.1 & 157 & \textbf{0.007} \\
MPH Table & 35.0 & 1412 & 0.569 & 49.0 & 1505 & 0.438 & 55.9 & 1551 & 0.572 \\
Learned CSF & \textbf{4.4} & 6435 & 5.173 & \textbf{2.2} & 6160 & 5.480 & \textbf{0.7} & 3199 & 5.733 \\
\midrule
\multicolumn{10}{c}{\emph{Zipfian}} \\
\midrule
AutoCSF & \textbf{4.6} & 98 & 0.025 & \textbf{2.3} & 95 & 0.014 & \textbf{0.8} & \textbf{78} & 0.007 \\
BCSF (Shibuya) & 4.9 & \textbf{93} & 0.026 & 2.4 & \textbf{90} & 0.013 & 0.8 & 94 & 0.008 \\
C++ Hash Table & 95.1 & 111 & \textbf{0.006} & 95.1 & 142 & \textbf{0.006} & 95.1 & 107 & \textbf{0.006} \\
MPH Table & 35.4 & 1195 & 0.569 & 49.2 & 1344 & 0.490 & 56.0 & 1323 & 0.470 \\
Learned CSF & 4.9 & 76811 & 32.064 & 2.4 & 39387 & 22.474 & 1.4 & 13530 & 13.284 \\
\midrule
\multicolumn{10}{c}{\emph{Unique}} \\
\midrule
AutoCSF & 26.3 & 116 & 0.067 & 10.6 & 103 & 0.032 & \textbf{2.7} & \textbf{79} & 0.012 \\
BCSF (Shibuya) & 27.3 & \textbf{101} & 0.057 & 11.0 & \textbf{93} & 0.026 & 2.7 & 88 & 0.012 \\
C++ Hash Table & 95.1 & 114 & \textbf{0.007} & 95.1 & 204 & \textbf{0.006} & 95.1 & 171 & \textbf{0.006} \\
MPH Table & 43.9 & 1363 & 0.496 & 52.2 & 1236 & 0.487 & 56.4 & 1390 & 0.472 \\
Learned CSF & \textbf{11.4} & 3997545 & 983.000 & \textbf{6.9} & 1431918 & 359.657 & 3.5 & 243241 & 85.772 \\
\bottomrule
\end{tabular}
\caption{Synthetic benchmark results ($N = 100{,}000$). Memory in bits/key (bpk), query latency in nanoseconds (ns), and construction time in seconds. \textbf{Bold} indicates best in each column.}
\label{tab:synthetic}
\end{table*}

\subsection{Genomics Benchmarks}

Table~\ref{tab:genomics} reports results on the three genomics datasets. The datasets span a wide range of skew: E.\ coli is highly skewed ($\alpha = 0.97$), C.\ elegans has moderate skew ($\alpha = 0.82$), and SRR has low skew ($\alpha = 0.20$). This diversity tests whether AutoCSF's decision criterion correctly adapts its strategy to the data.
\begin{table*}[t]
\centering
\begin{tabular}{lrrrrrrrrr}
\toprule
 & \multicolumn{3}{c}{E.\ coli ($N$=5.3M, $\alpha$=0.97)} & \multicolumn{3}{c}{SRR ($N$=9.8M, $\alpha$=0.20)} & \multicolumn{3}{c}{C.\ elegans ($N$=69.7M, $\alpha$=0.82)} \\
Method & bpk & query (ns) & build (s) & bpk & query (ns) & build (s) & bpk & query (ns) & build (s) \\
\midrule
AutoCSF & \textbf{0.31} & \textbf{92} & \textbf{0.2} & 4.21 & 443 & 2.3 & \textbf{1.26} & 466 & \textbf{6.6} \\
BCSF (Shibuya) & 0.35 & 126 & 0.3 & 4.24 & \textbf{355} & 2.3 & 1.39 & \textbf{274} & 7.0 \\
C++ Hash Table & 152.00 & 488 & 0.8 & 152.00 & 742 & \textbf{1.5} & 152.00 & 1034 & 15.2 \\
MPH Table & 11.55 & 1666 & 8.1 & 11.55 & 1716 & 13.6 & 11.55 & 2054 & 136.3 \\
Learned CSF & 0.32 & 7843 & 161.0 & \textbf{3.49} & 31811 & 676.9 & 1.68 & 29558 & 4573.5 \\
\bottomrule
\end{tabular}
\caption{Genomics benchmark results. Memory in bits/key (bpk), query latency in nanoseconds (ns), and construction time in seconds. \textbf{Bold} indicates best in each column.}
\label{tab:genomics}
\end{table*}

On E.\ coli, where $\alpha = 0.97$, AutoCSF achieves 0.31~bpk, less than one-third of a bit per key to index 5.3 million $k$-mers. This is $37\times$ smaller than MPH Table (11.55~bpk). Learned CSF is competitive in memory (0.32~bpk) but requires 161 seconds to build versus AutoCSF's 0.2 seconds ($805\times$ slower) and has $85\times$ higher query latency (7{,}843~ns vs.\ 92~ns).

The SRR dataset ($\alpha = 0.20$) represents a case where filter augmentation is not beneficial. AutoCSF correctly defaults to a plain CSF, achieving 4.21~bpk with 443~ns query latency. Learned CSF achieves the best memory (3.49~bpk) but at a cost: 677 seconds to build ($294\times$ slower than AutoCSF) and $72\times$ higher query latency (31{,}811~ns vs.\ 443~ns).

On C.\ elegans ($n = 69.7$M, $\alpha = 0.82$), AutoCSF scales to the largest dataset while maintaining its advantages: 1.26~bpk with 466~ns latency and a 6.6-second build time. For context, this means the entire 70-million-entry $k$-mer count table is indexed in approximately 11~MB of memory. MPH Table would require 100.6~MB for the same data. Learned CSF uses 1.68~bpk (33\% more memory than AutoCSF) with 29{,}558~ns latency ($63\times$ slower) and requires over 76 minutes to build.

\subsection{Discussion}

Across all experiments, a consistent picture emerges. AutoCSF dominates the Pareto frontier of memory versus latency, especially when the majority fraction $\alpha$ is moderate to high ($\gtrsim 0.7$). The margin over heuristic construction is modest (up to 11\% in memory) but systematic, confirming that the theory-guided parameter selection from Section~4 translates to measurable improvements in practice over the heuristic approach.

The comparison with Learned CSF is particularly instructive. Learned CSF can achieve tighter compression than AutoCSF in several regimes: it consistently uses less memory on Uniform-100, and on Unique it achieves substantially lower bits per key at low-to-moderate $\alpha$ (e.g., 11.4 vs.\ 26.3~bpk at $\alpha = 0.5$), demonstrating that a trained model can exploit distributional structure beyond what a static Huffman code captures. However, these memory savings come at 2--4 orders of magnitude higher query latency due to per-key neural network inference, and construction times ranging from seconds to over 16 minutes due to model training. On Zipfian distributions, where the value distribution follows a smoother power law, AutoCSF achieves lower memory than Learned CSF at all $\alpha$ values without any training overhead. These results suggest that learned approaches offer genuine compression advantages in specific regimes, but for practical deployments where query latency and build time are important, AutoCSF provides a more favorable overall tradeoff.

\section{Conclusion}

We presented AutoCSF, a novel, theoretically principled algorithm that determines both \emph{when} to augment a compressed static function with an approximate set membership data structure and \emph{how} to optimally set the parameters for this index to enable fast key-value lookups with a small memory footprint. While the idea of augmenting a CSF with a filter was previously proposed in the literature \cite{shibuya2022space}, the existing state of the art relied on data-driven heuristics, which, as we showed in this paper, can provide misleading guidance in practice. The core innovation in AutoCSF is the development of novel lower and upper bounds on the \emph{difference} in cost between a CSF and filter-augmented CSF, using mathematical tools from information theory. We validate that our theoretical results closely align with empirical observations and provide accurate guidance in parameter selection for practitioners. We also compare AutoCSF to other succinct retrieval systems in the literature, such as BCSF and Learned CSF, on both synthetic and real genomics workloads where we find that AutoCSF provides a Pareto-optimal tradeoff in query latency and memory footprint while maintaining low index build times. 

\section*{Acknowledgements}

We thank Bryce Kille for valuable discussions on early drafts of this work. \\

\noindent This material is based upon work supported by the U.S. National Science Foundation under Grant
No. 2313998. Any opinions, findings, and conclusions or recommendations expressed in this material are those of the author(s) and do not necessarily reflect the views of the U.S. National Science
Foundation.


\bibliographystyle{ACM-Reference-Format}
\bibliography{sample}

\end{document}